\documentclass[conference]{IEEEtran}

\usepackage{amsfonts}
\usepackage{amsmath}
\usepackage{amssymb}
\usepackage{lscape}
\usepackage{epsf}
\usepackage{graphicx}
\usepackage{psfrag}

\newtheorem{theorem}{Theorem}%{\indent Theorem}[section]
\newtheorem{lemma}{Lemma}%[theorem]{\indent Lemma}
\newtheorem{proposition}{Proposition}%[theorem]{\indent Proposition}
{\indent Claim}
\newtheorem{EXAMPLE}{Example}%{\indent Example}[section]
\newtheorem{definition}{Definition}%{\indent Definition}%[section]

\newcommand{\code}{{\mathcal{C}}}
\newcommand{\graph}{{\mathcal{G}}}
\newcommand{\cA}{{\mathcal{A}}}

\newcommand{\cV}{{\mathcal{V}}}

\newcommand{\cE}{{\mathcal{E}}}

\newcommand{\cB}{{\mathcal{B}}}
\newcommand{\cI}{{\mathcal{I}}}
\newcommand{\cJ}{{\mathcal{J}}}
\newcommand{\cM}{{\mathcal{M}}}
\newcommand{\cN}{{\mathcal{N}}}
\newcommand{\cK}{{\mathcal{K}}}

\newcommand{\cS}{{\mathcal{S}}}

\newcommand{\cU}{{\mathcal{U}}}

\newcommand{\coeff}{{\mbox{Coeff }}}

\newcommand{\bldalpha}{{\mbox{\boldmath $\alpha$}}}

\newcommand{\bldc}{{\mbox{\boldmath $c$}}}

\newcommand{\bldnu}{{\mbox{\boldmath $\nu$}}}

\newcommand{\bldp}{{\mbox{\boldmath $p$}}}

\newcommand{\bldq}{{\mbox{\boldmath $q$}}}
\newcommand{\bldqq}{{\mbox{\scriptsize \boldmath $q$}}}

\newcommand{\bldt}{{\mbox{\boldmath $t$}}}

\newcommand{\bldv}{{\mbox{\boldmath $v$}}}
\newcommand{\bldvv}{{\mbox{\scriptsize \boldmath $v$}}}
\newcommand{\bldu}{{\mbox{\boldmath $u$}}}
\newcommand{\blduu}{{\mbox{\scriptsize \boldmath $u$}}}
\newcommand{\bldw}{{\mbox{\boldmath $w$}}}
\newcommand{\bldww}{{\mbox{\scriptsize \boldmath $w$}}}
\newcommand{\bldx}{{\mbox{\boldmath $x$}}}

\newcommand{\bldz}{{\mbox{\boldmath $z$}}}
\newcommand{\bldxi}{{\mbox{\boldmath $\xi$}}}

\newcommand{\zeros}{{\mbox{\boldmath $0$}}}

\newcommand{\bldzero}{{\mbox{\boldmath $0$}}}

    \def\squarebox#1{\hbox to #1{\hfill\vbox to #1{\vfill}}}

\newlength{\Algwidth}

\title{Exposing Pseudoweight Layers in Regular LDPC Code Ensembles}

\author{
\authorblockN{Mark F. Flanagan}
\authorblockA{University College Dublin \\
Belfield, Dublin 4, Ireland  \\
Email: mark.flanagan@ieee.org}
}

\begin{document}
\maketitle

\begin{abstract}
A solution is presented for the asymptotic growth rate of the AWGN-pseudoweight distribution of regular low-density parity-check (LDPC) code ensembles for a selected graph cover degree $M \ge 1$. The evaluation of the growth rate requires solution of a system of $2M+1$ nonlinear equations in $2M+1$ unknowns. Simulation results for the pseudoweight distribution of two regular LDPC code ensembles are presented for graph covers of low degree.
\end{abstract}
%\begin{keywords}
%Doubly-generalized LDPC codes,
%irregular code ensembles,
%weight distribution. 
%\end{keywords}

\section{Introduction}

In classical coding theory, the \emph{weight distribution} of a code is a useful tool for measuring a linear code's performance under maximum likelihood (ML) decoding. For codes decoded using modern high-performance suboptimal decoding algorithms such as sum-product (SP) or linear-programming (LP) decoding, the \emph{pseudoweight} is the appropriate analog of the codeword weight. There are different definitions of pseudoweight for different channels; one of primary importance is the \emph{additive white Gaussian noise} (AWGN) pseudoweight. The pseudoweight distribution considers all codewords in all codes derived from finite covers of the Tanner graph, which compete with the codewords to be the best SP decoding solution. The set of pseudocodewords has a succinct characterization in terms of the so-called \emph{fundamental polytope} or equivalently, the \emph{fundamental cone} \cite{FKKR,KV-characterization,KV-IEEE-IT}. Also, pseudocodewords arising from finite covers of the Tanner graph were shown to be equivalent to those responsible for failure of LP decoding \cite{Feldman-thesis, Feldman}. While much of the existing work in this area is concerned with performance characterization of particular codes, the performance of \emph{ensembles} of \emph{low-density parity-check} (LDPC) codes \cite{gallager63:low-density} is also of interest. 

In \cite{di06:weight}, the growth rate of the weight distribution of irregular LDPC codes was derived, and a numerical technique was presented for its approximate evaluation. It was shown in \cite[Corollary 50]{KV-IEEE-IT} 
%(using a bound on the AWGN pseudoweight contained in \cite{KV-lower-bounds}) 
that $(j,k)$-regular ensembles with $j \ge 3$ have a ratio of minimum AWGN-pseudoweight to block length $n$ which decreases to zero asymptotically as $n \rightarrow \infty$. Apart from this result, to the author's knowledge no \emph{ensemble} results exist in the literature concerning AWGN-pseudoweight. In this paper, we make a first step in this direction. We define the degree-$M$ \emph{pseudoweight enumerating function} of a linear block code, and use this concept to find an expression for the growth rate of the AWGN-pseudoweight of regular LDPC code ensembles. We also present simulation results for the $(3,6)$-regular and $(4,8)$-regular LDPC code ensembles.

%%%%%%%%%%%%%%%%%%%%%%%%%%%%%%%%%%%%%%%%%%%%%%
%%%%%%%%%%%%%%%%%%%%%%%%%%%%%%%%%%%%%%%%%%%%%%

\section{Preliminaries and Notation}\label{section:irregular_D_GLDPC}
We begin by providing some general settings and definitions. For $\bldu = \left( u_1 \; u_2 \; \cdots \; u_M \right)$, we denote the multinomial coefficient by
\[
\binom{k}{\bldu} = \binom{k}{u_1 \; u_2 \; \cdots \; u_M} = \frac{k!}{\left( k-\sum_{r=1}^{M} u_r \right)! \prod_{r=1}^{M} u_r !} \; .
\] 
For $\bldalpha = \left( \alpha_1 \; \alpha_2 \; \cdots \; \alpha_M \right) \in \mathbb{R}^M$ with $\alpha_r \ge 0$ for each $r=1,2,\cdots,M$ and $\sum_{r=1}^{M} \alpha_r \le 1$, we denote the multivariate entropy function by 
\begin{equation}
h(\bldalpha)= - \sum_{r=1}^{M} \alpha_r \log \alpha_r - \left( 1-\sum_{r=1}^{M} \alpha_r \right) \log \left( 1-\sum_{r=1}^{M} \alpha_r \right) \; .
\label{eq:entropy_function}
\end{equation}
All logarithms in the paper are to the base $e$.

Let $\code$ be a linear block code of length $n$ over the binary field $\mathbb{F}_2$, defined by
\begin{equation}
\code = \{ \bldc \in \mathbb{F}_2^n \; : \; \bldc \mathbf{H}^T = \bldzero \in \mathbb{F}_2^m \}
\label{eq:code_definition}
\end{equation}
where $\mathbf{H} = (H_{j,i})$ is an $m \times n$ matrix over $\mathbb{F}_2$ called the \emph{parity-check matrix} of the code $\code$. Also denote $\cI = \{ 1,2,\cdots,n \}$, $\cJ = \{ 1,2,\cdots,m \}$ and for each $j \in \cJ$
\[
\cI_j = \{ i \in \cI \; : \; H_{j,i} = 1\} \; .
\]

The \emph{Tanner graph} of a linear block code $\code$ over $\mathbb{F}_2$ with parity-check matrix $\mathbf{H}$ is an equivalent characterization of $\mathbf{H}$. The Tanner graph $\graph = (\cV, \cE)$ has vertex set $\cV = \{u_1, u_2, \cdots, u_n \} \cup \{v_1, v_2, \cdots, v_m \}$, and there is an edge between $u_i$ and $v_j$ if and only if $H_{j,i} = 1$. We denote by $\cN(v)$ the set of neighbors of a vertex $v\in\cV$.

We next define what is meant by a finite cover of a Tanner graph.
\begin{definition}
(\cite{KV-characterization})
A graph $\tilde{\graph} = (\tilde{\cV}, \tilde{\cE})$ is a \emph{finite cover} of the Tanner graph $\graph = (\cV, \cE)$ 
if there exists a mapping $\Pi: \tilde{\cV} \longrightarrow \cV$ which is a graph homomorphism
($\Pi$ takes adjacent vertices of $\tilde{\graph}$ to adjacent vertices of $\graph$), such that 
for every vertex $v \in \cV$ and every $\tilde{v} \in \Pi^{-1}(v)$, the neighborhood $\cN(\tilde{v})$ 
of $\tilde{v}$ is mapped bijectively to $\cN(v)$. 
\end{definition}
\medskip
\begin{definition}
(\cite{KV-characterization})
A cover of the graph $\graph$ is said to have degree $M$, where $M$ is a positive integer, if $|\Pi^{-1}(v)| = M$
for every vertex $v \in \cV$. We refer to such a cover graph as an $\mit{M}$\emph{-cover} of $\graph$.  
\end{definition}
\medskip
Let $\tilde{\graph} = (\tilde{\cV}, \tilde{\cE})$ be an $M$-cover   
of the Tanner graph $\graph = (\cV, \cE)$ representing 
the code $\code$ with parity-check matrix $\mathbf{H}$. 
The vertices in the set $\Pi^{-1} (u_i)$ are called \emph{copies} of $u_i$ and are denoted $\{ u_{i,1}, u_{i,2}, \cdots, u_{i,M} \}$, where $i\in\cI$. Similarly, the vertices in the set $\Pi^{-1} (v_j)$ are called \emph{copies} of $v_j$ and are denoted $\{ v_{j,1}, v_{j,2}, \cdots, v_{j,M} \}$, where $j\in\cJ$. 

Less formally, given a code $\code$ with parity check matrix $\mathbf{H}$ and corresponding Tanner graph $\graph$, an $M$-cover of $\graph$ is a graph whose vertex set consists of $M$ copies of each vertex $u_i$ and $M$ copies of each vertex $v_j$, such that for each $j\in\cJ$, $i\in\cI_j$, the $M$ copies of $u_i$ and the $M$ copies of $v_j$ are connected in an arbitrary one-to-one fashion.

For any $M\ge 1$, an $M$-\emph{cover codeword} is a labelling of vertices of the $M$-cover graph with values from $\mathbb{F}_2$ such that all parity checks are satisfied. We denote the label of $u_{i,r}$ by $p_{i,r}$ for each $i\in\cI$, $r = 1, 2, \cdots, M$, and we may then write the $M$-cover codeword in vector form as 
\begin{multline*} 
\bldp = \big( p_{1,1}, p_{1,2}, \cdots, p_{1,M}, p_{2,1}, p_{2,2}, \\
\cdots, p_{2,M}, \cdots, p_{n,1}, p_{n,2}, \cdots, p_{n, M} \big) \; .
\end{multline*}
It is easily seen that $\bldp$ belongs to a linear code $\tilde{\code}$ of length $Mn$ over $\mathbb{F}_2$, 
defined by an $Mm \times Mn$ parity-check matrix $\tilde{\mathbf{H}}$. To construct $\tilde{\mathbf{H}}$, for
$1 \le i^*,j^* \le M$ and $i \in \cI$, $j \in \cJ$, we let  
$i' = (i-1) M + i^*, j' = (j-1) M + j^*$, and
\[
\tilde{H}_{j',i'} = \left\{ \begin{array}{cl}
1 & \mbox{if } u_{i,i^*} \in \cN(v_{j,j^*}) \\
0 & \mbox{otherwise.} 
\end{array} \right.
\]
It may be seen that $\tilde{\graph}$ is the Tanner graph of the code $\tilde{\code}$ corresponding to the parity-check matrix $\tilde{\mathbf{H}}$.

We next define the concept of \emph{pseudocodeword} as follows.
\medskip
\begin{definition}
Let $\code$ be a linear code of length $n$ with parity-check matrix $\mathbf{H}$. For any positive integer $M$, a vector $\bldz = \left( z_1 \; z_2 \; \cdots \; z_n \right)$ of length $n$ with nonnegative integer entries is said to be a \emph{degree}-$M$ \emph{pseudocodeword} of the code $\code$ if and only if there exists an $M$-cover codeword $\bldp$ with 
\[
z_i = \left| \{ r \in \{ 1,2,\cdots,M \} \; : \; p_{i,r} = 1 \} \right| 
\]
for all $i \in \cI$.
\end{definition}
\medskip
\begin{definition}
The \emph{pseudoweight} of a degree-$M$ pseudocodeword $\bldz$ of the code $\code$ is equal to the vector $\bldu = \left( u_1 \; u_2 \; \cdots \; u_M \right)$, where $\bldz$ has $u_r$ entries equal to $r$ for each $r=1,2,\cdots,M$. 
\end{definition}
\medskip
Note that the pseudoweight as defined here corresponds to the ``type'' of a pseudocodeword in the notation of \cite{Smarandache_Vontobel}. Note also that this notion of pseudoweight is applicable to different channels such as the additive white Gaussian noise (AWGN) channel, binary symmetric channel (BSC) or binary erasure channel (BEC). The AWGN-pseudoweight of a pseudocodeword $\bldz$ of length $n$ is defined by \cite{Wiberg}
\begin{equation}
w(\bldz) = \frac{\left( \sum_{i \in \cI} z_i \right) ^2}{\sum_{i \in \cI} z_i^2}
\label{eq:AWGN-pseudoweight}
\end{equation}
and its BSC-pseudoweight and BEC-pseudoweight are defined in \cite{FKKR} (see also \cite[Section 6]{KV-IEEE-IT}).
\medskip
\begin{definition}[\cite{KV-characterization}]
The \emph{fundamental cone} $\cK(\mathbf{H})$ of the $m \times n$ parity-check matrix $\mathbf{H}$ is equal to the set of vectors $\bldnu = (\nu_1 \; \nu_2 \; \cdots \; \nu_n) \in \mathbb{R}^n$ such that $\nu_i \ge 0$ for all $i \in \cI$, and
\[
\sum_{\ell \in \cI_j\backslash \{ i \}} \nu_{\ell} \ge \nu_i \quad \forall j \in \cJ, i \in \cI_j \; .
\]
\label{def:cone}
\end{definition}
\medskip
In \cite{KV-characterization}, it was shown that if $\code$ is a binary linear code with an $m \times n$ parity-check matrix $\mathbf{H}$, then a length-$n$ integer vector $\bldz$ is a pseudocodeword\footnote{Note that the object we call a ``pseudocodeword'' was called an ``unscaled pseudocodeword'' in \cite{KV-characterization}.} of $\mathbf{H}$ if and only if $\bldz \in \cK(\mathbf{H})$ and
\begin{equation}
\bldz \mathbf{H}^T \equiv \zeros \: \: (\mbox{mod } 2)
\label{eq:pcw_condition_2_parity}
\end{equation}
where $\cK(\mathbf{H})$ denotes the fundamental cone of $\mathbf{H}$, and the matrix $\mathbf{H}$ in (\ref{eq:pcw_condition_2_parity}) is interpreted over the integers.
\medskip

We next define the concept of pseudoweight enumerating function of a block code.
\medskip
\begin{definition}
The degree-$M$ \emph{pseudoweight enumerating function} (PWEF) of a block code $\code$ of length $n$ is equal to 
\[
B^{(M)}(\bldx) = \sum_{\blduu} B^{(M)}_{\blduu} x_1^{u_1} x_2^{u_2} \cdots x_M^{u_M}
\]
where $\bldx = \left( x_1 \; x_2 \; \cdots \; x_M \right)$, $\bldu = \left( u_1 \; u_2 \; \cdots \; u_M \right)$ and $B^{(M)}_{\blduu}$ denotes the number of degree-$M$ pseudocodewords\footnote{Note that this count does \emph{not} consider the multiplicity of $M$-cover codewords corresponding to a particular pseudocodeword.} of the code with pseudoweight $\bldu$.
\end{definition}
\medskip
\begin{proposition}
\label{prop:B}
The degree-$M$ PWEF of the single parity-check (SPC) code of length $k$ is 
\begin{equation}
B^{(M)}(\bldx) = \frac{1}{2} \left[ \left( P^{(M)}(\bldx)\right) ^k + \left( Q^{(M)}(\bldx)\right) ^k \right] - T^{(M)}(\bldx)
\label{eq:B_definition}
\end{equation}
where
%\begin{equation}
%P^{(M)}(\bldx) \triangleq 1 + \sum_{r=1}^{M} x_{r} \; ,
%\label{eq:P_definition}
%\end{equation} 
%\begin{equation}
%Q^{(M)}(\bldx) \triangleq 1 + \sum_{r=1}^{M} (-1)^r x_{r} \; ,
%\label{eq:Q_definition}
%\end{equation}
$P^{(M)}(\bldx) \triangleq 1 + \sum_{r=1}^{M} x_{r}$, $Q^{(M)}(\bldx) \triangleq 1 + \sum_{r=1}^{M} (-1)^r x_{r}$, $T^{(1)}(\bldx) = 0$, and for $M>1$
\begin{multline}
T^{(M)}(\bldx) = T^{(M-1)}(\bldx) \\
+ x_M \left( \sum_{\bldvv \in \cU^{M-1}} \binom{k}{1 \; v_1 \; v_2 \; \cdots \; v_{M-1}} x_1^{v_1} x_2^{v_2} \cdots x_{M-1}^{v_{M-1}}\right)
\label{eq:T_recursion}
\end{multline}
where the set $\cU^{M-1}$ in (\ref{eq:T_recursion}) is the set of integer vectors $\bldv = (v_1 \; v_2 \; \cdots \; v_{M-1})$ satisfying $v_r \ge 0$ for all $r=1,2,\cdots,M-1$, $\sum_{r=1}^{M-1} r v_r < M$ and where $\sum_{r \; \mathrm{odd}} v_r + M$ is even. 
\end{proposition}
\medskip
\begin{proof}
In this case $\mathbf{H}$ is a length-$k$ row vector of ones, so we have
\[
B^{(M)}_{\blduu} =  \left\{ \begin{array}{cl}
\binom{k}{\blduu} & \mbox{if } \bldu \in \cS \\
0 & \mbox{otherwise,} 
\end{array} \right.
\]
where (using Definition \ref{def:cone}) $\cS$ is the set of integer vectors $\bldu$ which satisfy 
\begin{description}
\item[($\cS$--$1$)]
\begin{equation}
u_r \ge 0 \:\: \forall r=1,2,\cdots, M \:\: ; \:\: \sum_{r=1}^{M} u_r \le k, 
\label{eq:SPC_PCW_ineq}
\end{equation}
\item[($\cS$--$2$)]
\begin{equation}
\sum_{r \; \mathrm{odd}} u_r \mbox{ is even, and}
\label{eq:SPC_PCW_1}
\end{equation}
\item[($\cS$--$3$)] 
If there exists $c \in \{ 1,2,\cdots,M \}$ with $u_c=1$ and $u_r=0$ for $c < r \le M$, then
\begin{equation}
c \le \sum_{r=1}^{c-1} r u_r \; .
\label{eq:SPC_PCW_2}
\end{equation}
\end{description}
It is straightforward to check that in (\ref{eq:B_definition}), the term 
\[
\frac{1}{2} \left[ \left( P^{(M)}(\bldx)\right) ^k + \left( Q^{(M)}(\bldx)\right) ^k \right]
\]
takes into account all integer vectors $\bldu$ which satisfy conditions ($\cS$--$1$) and ($\cS$--$2$), and the term $T^{(M)}(\bldx)$ takes into account all those which violate the condition ($\cS$--$3$).
\end{proof}
\medskip
In particular $T^{(2)}(\bldx) = k x_2$, $T^{(3)}(\bldx) = k x_2 + k (k-1) x_1 x_3$, and $T^{(4)}(\bldx) = k x_2 + k (k-1) x_1 x_3 + k x_4 + k(k-1) x_2 x_4 + \frac{1}{2} k(k-1)(k-2) x_1^2 x_4$. Also note that
\begin{multline}
\frac{\partial B^{(M)}(\bldx)}{\partial x_r} = \frac{k}{2} \left[ \left(P^{(M)}(\bldx) \right)^{k-1} + (-1)^r \left( Q^{(M)}(\bldx) \right)^{k-1} \right] \nonumber \\
- \frac{\partial T^{(M)}(\bldx)}{\partial x_r}
\label{eq:B_derivative}
\end{multline}
for $r=1,2,\cdots,M$.

%%%%%%%%%%%%%%%%%%%%%%%%%%%%%%%%%%%%%%%%%%%%%%
%%%%%%%%%%%%%%%%%%%%%%%%%%%%%%%%%%%%%%%%%%%%%%

\section{Growth Rate of the AWGN-Pseudoweight Distribution of the Regular LDPC Code ensemble}\label{sec:growth_rate}
For a positive integer $n$, we define a $(j,k)$-regular LDPC code ensemble $\cM_n$ as follows. The Tanner graph of an LDPC code from the ensemble consists of $n$ variable nodes of degree $j$, and $m = nj/k$ check nodes of degree $k$. The variable and check node sockets are connected by a permutation on the $E=nj=mk$ edges of the graph, each permutation being equiprobable.

The concept of degree-$M$ assignment is defined next. This definition is a generalization of the definition of assignment in \cite{di06:weight} (the definition in \cite{di06:weight} corresponds to that of a degree-$1$ assignment).
\medskip
\begin{definition}
A \emph{degree-$M$ assignment} is a labelling of the edges of the Tanner graph with numbers from the set $\{ 0,1,2, \cdots,M \}$. An assignment is said to have \emph{pseudoweight} $\bldt = \left( t_1 \; t_2 \; \cdots \; t_M \right)$ if $t_r$ edges are labelled $r$ for each $r=1,2,\cdots,M$. An assignment is said to be $M$-\emph{check-valid} if according to this labelling, every check node recognizes a valid local degree-$M$ pseudocodeword.
\end{definition} 
\medskip
For any positive integer $M$, the growth rate of the degree-$M$ AWGN-pseudoweight distribution of the $(j,k)$-regular LDPC code ensemble sequence $\{ \cM_n \}$ is defined by 
\begin{equation}
G_M(\alpha) \triangleq \lim_{n\rightarrow \infty} \frac{1}{n} \log \mathbb{E}_{\cM_n} \left[ N^{(M)}_{\alpha n} \right]
\label{eq:growth_rate_result}
\end{equation}
where $\mathbb{E}_{\cM_n}$ denotes the expectation operator over the ensemble $\cM_n$, and $N^{(M)}_{w}$ denotes the number of degree-$M$ pseudocodewords of AWGN-pseudoweight $w$ of a randomly chosen LDPC code in the ensemble $\cM_n$. The limit in (\ref{eq:growth_rate_result}) assumes the inclusion of only those positive integers $n$ for which $\alpha n \in \mathbb{Z}$ and $\mathbb{E}_{\cM_n} [ N_{\alpha n} ]$ is positive (i.e., where the expression whose limit we seek is well defined).

We next define a notion of \emph{asymptotic goodness} of an LDPC code ensemble sequence.
\medskip
\begin{definition}
For each $M \ge 1$, let $G_M(\alpha)$ be the growth rate of the degree-$M$ pseudoweight distribution of an LDPC code ensemble sequence, and let $\alpha_M^*=\inf \{ \alpha > 0 \; | \; G_M(\alpha)\geq 0 \}$. The ensemble sequence is said to be \emph{asymptotically good} if and only if $\inf_{M \ge 1} \{\alpha_M^* \} > 0$.
\end{definition}
\medskip
The following theorem constitutes the main result of the paper.
\medskip
\begin{theorem}
The growth rate of the degree-$M$ pseudoweight distribution of the $(j,k)$-regular LDPC code ensemble sequence $\{ \cM_n \}$ is given by
\begin{equation}
G_M(\alpha) = \frac{j}{k} \log B^{(M)}(\bldx_{0}) - j \sum_{i=1}^{M} q_i \log x_{0,i} - (j-1) h(\bldq)
\label{eq:growth_rate_polynomial_general}
\end{equation}
where $\bldx_0 = \left(x_{0,1} \; x_{0,2} \; \cdots \; x_{0,M} \right)$, $\bldq = \left( q_1 \; q_2 \; \cdots \; q_M \right)$ and $\lambda$ are the solutions to the system of $2M+1$ equations in $2M+1$ unknowns\footnote{Note that $B^{(M)}(\bldx)$ is given by Proposition \ref{prop:B}.}
\begin{equation}
\label{eq:bldx0_q_eqn}
x_{0,r} \frac{\partial B^{(M)}(\bldx_0)}{\partial x_{0,r}} = k q_r B^{(M)}(\bldx_0)
\end{equation}
for each $r=1,2,\cdots,M$,
\begin{multline}
(j-1) \log \left[ \frac{q_r}{1-\sum_{s=1}^{M} q_s} \right] - j \log x_{0,r} \\
= \lambda \left( 2r \sum_{s=1}^{M} s q_s - \alpha r^2 \right)
\label{eq:lagrange_mult_theorem}
\end{multline}
for each $r=1,2,\cdots,M$, and
\begin{equation}
g(\bldq) = \left( \sum_{r=1}^{M} r q_r \right) ^2 - \alpha \sum_{r=1}^{M} r^2 q_r = 0
\label{eq:g_definition}
\end{equation}
satisfying $x_{0,r}>0$ and $q_r>0$ for each $r=1,2,\cdots,M$.
\label{thm:growth_rate} 
\end{theorem}
\begin{proof}
Consider a degree-$M$ pseudocodeword $\bldz = \left( z_1 \; z_2 \; \cdots \; z_n \right)$ with pseudoweight $\bldq n$, where $\bldq = \left( q_1 \; q_2 \; \cdots \; q_M \right)$. This pseudocodeword naturally induces a degree-$M$ assignment of pseudoweight $j \bldq n$. Using (\ref{eq:AWGN-pseudoweight}), the AWGN-pseudoweight of $\bldz$ may be written as $w(\bldz) = \alpha n$ where
\begin{equation}
\alpha = \frac{\left( \sum_{r=1}^{M} r q_r \right) ^2}{\sum_{r=1}^{M} r^2 q_r} \; .
\label{eq:alpha}
\end{equation}
Rearranging (\ref{eq:alpha}), and defining $g(\bldq)$ appropriately, yields (\ref{eq:g_definition}). The expected number of degree-$M$ pseudocodewords of pseudoweight $\bldq n$ is then 
\begin{equation}
\mathbb{E}_{\cM_n} \left[ N^{(M)}(\bldq) \right] = \binom{n}{\bldq n} \cdot P^{(M)}_{\mbox{\scriptsize c-valid}}(j \bldq) \; ,
\label{eq:Expected_over_q}
\end{equation}
where $P^{(M)}_{\mbox{\scriptsize c-valid}}(\bldalpha)$ represents the probability that a randomly chosen degree-$M$ assignment with pseudoweight $\bldalpha n$ is $M$-check-valid. This probability is given by
\begin{equation}
P^{(M)}_{\mbox{\scriptsize c-valid}}(j \bldq) = N_c^{(M)}(j \bldq) \Big/ \binom{j n}{j \bldq n} \; ,
\label{eq:Prob_valid_assign}
\end{equation}
where $N_c^{(M)}(\bldalpha)$ denotes the number of $M$-check-valid degree-$M$ assignments of pseudoweight $\bldalpha n$. The numerator of (\ref{eq:Prob_valid_assign}) may be written as
\footnote{Here we use the following result on multivariate generating functions. Let $a_{\blduu}$ be the number of ways of obtaining an outcome $\bldu = (u_1 \; u_2 \; \cdots \; u_M)\in\mathbb{Z}^M$ in experiment $\cA$, and let $b_{\bldvv}$ be the number of ways of obtaining an outcome $\bldv = (v_1 \; v_2 \; \cdots \; v_M)\in\mathbb{Z}^M$ in experiment $\cB$. Also let $c_{\bldww}$ be the number of ways of obtaining an outcome $(\bldu,\bldv)$ in the combined experiment $(\cA, \cB)$ such that $\bldu + \bldv = \bldw$. Denoting $\bldx = (x_1 \; x_2 \; \cdots \; x_M)$, the generating functions $A(\bldx)=\sum_{\blduu} a_{\blduu} x_1^{u_1} x_2^{u_2} \cdots x_M^{u_M}$, $B(\bldx)=\sum_{\bldvv} b_{\bldvv} x_1^{v_1} x_2^{v_2} \cdots x_M^{v_M}$ and $C(\bldx)=\sum_{\bldww} c_{\bldww} x_1^{w_1} x_2^{w_2} \cdots x_M^{w_M}$ are related by $C(\bldx) = A(\bldx) B(\bldx)$.} 
\[
N_{c}^{(M)}(j \bldq) = \coeff \left[ \left( B^{(M)}(\bldx) \right) ^{m}, x_1^{j q_1 n} x_2^{j q_2 n} \cdots x_M^{j q_M n} \right] \; .
\]
We next make use of the following result from \cite[Theorem 2]{Burshtein_Miller}:
\medskip
\begin{lemma}
Let $R(\bldx)$ denote a multivariate polynomial with nonnegative coefficients. For a fixed vector of positive rational numbers $\bldxi = \left( \xi_1 \; \xi_2 \; \cdots \; \xi_M \right)$, consider the set of positive integers $\ell$ such that $\xi_r \ell \in \mathbb{Z}$ for each $r=1,2,\cdots,M$ and $\coeff ( \{ R(\bldx) \} ^{\ell}, x_1^{\xi_1 \ell} x_2^{\xi_2 \ell} \cdots x_M^{\xi_M \ell}) > 0$. Then either this set is empty, or it has infinite cardinality; if $t$ is one such $\ell$, then so is $jt$ for every positive integer $j$. In the latter case, the following limit is well defined and exists:
\begin{multline}
\lim_{\ell\rightarrow \infty} \frac{1}{\ell} \log \coeff \left[ \left( R(\bldx) \right)^{\ell}, x_1^{\xi_1 \ell} x_2^{\xi_2 \ell} \cdots x_M^{\xi_M \ell} \right] \\
= \log R(\bldx_0) - \sum_{r=1}^{M} \xi_r \log x_{0,r}
\end{multline}
where $\bldx_0 = \left(x_{0,1} \; x_{0,2} \; \cdots \; x_{0,M} \right)$ is the unique positive real solution to the system of equations
\begin{equation}
\label{eq:bldx_soln}
x_{0,r} \frac{\partial R (\bldx_0)}{\partial x_{0,r}} = \xi_r R(\bldx_0) 
\end{equation}
for each $r=1,2,\cdots,M$.
\label{lemma:optimization_1D}
\end{lemma} 
\medskip
Applying this lemma by substituting $R(\bldx) = B^{(M)}(\bldx)$, $\ell=m=nj/k$ and $\xi_r = k q_r$, we obtain that as $n \rightarrow \infty$
\begin{equation}
N_{c}^{(M)}(j \bldq) \rightarrow \exp \left\{ n \left( \frac{j}{k} \log B^{(M)}(\bldx_0) - j \sum_{r=1}^{M} q_r \log x_{0,r} \right) \right\} 
\label{eq:Nct_epsilon_mid}
\end{equation}
where $\bldx_0 = \left(x_{0,1} \; x_{0,2} \; \cdots \; x_{0,M} \right)$ is the unique positive real solution to the system given by (\ref{eq:bldx0_q_eqn}) for each $r=1,2,\cdots,M$. Note that (\ref{eq:bldx0_q_eqn}) provides an implicit definition of $\bldx_0$ as a function of $\bldq$.

Using Stirling's formula, the multinomial coefficients in~(\ref{eq:Expected_over_q}) and~(\ref{eq:Prob_valid_assign}) may be approximated as $n \rightarrow \infty$ as
\[
\binom{n}{\bldq n} \rightarrow \exp \left\{ n h(\bldq) \right\} \; ; \; \binom{j n}{j \bldq n} \rightarrow \exp \left\{ n j h(\bldq) \right\} \; .
\]

Therefore as $n \rightarrow \infty$
\begin{multline}
\mathbb{E}_{\cM_n} \left[ N^{(M)}(\bldq) \right] \rightarrow \exp \Bigg\{ n \Bigg[ \frac{j}{k} \log B^{(M)}(\bldx_0) \\
- j \sum_{r=1}^{M} q_r \log x_{0,r} - (j-1) h(\bldq) \Bigg] \Bigg\}
\end{multline}

The expected number of degree-$M$ pseudocodewords with AWGN-pseudoweight $\alpha n$ is equal to the sum of the numbers of degree-$M$ pseudocodewords with pseudoweight $\bldq$ satisfying (\ref{eq:g_definition}), i.e.
\[
\mathbb{E}_{\cM_n} \left[ N^{(M)}_{\alpha n} \right] = \sum_{\bldqq \; : \; g(\bldqq)=0} \mathbb{E}_{\cM_n} \left[ N^{(M)}(\bldq) \right] \; .
\]
Note that the asymptotic expression as $n\rightarrow \infty$ is dominated by that $\bldq$ satisfying (\ref{eq:g_definition}) which maximizes the argument of the exponential function\footnote{Observe that as $n\rightarrow \infty$, $\sum_t \exp ( n Z_t ) \rightarrow \exp ( n \max_t \{Z_t\} )$}. Therefore as $n\rightarrow \infty$
\begin{equation}
\mathbb{E}_{\cM_n} \left[ N^{(M)}_{\alpha n} \right] \rightarrow \exp \left\{ n \left( \max_{\bldqq \; : \; g(\bldqq)=0} f(\bldq) \right) \right\}
\label{eq:optimization_for_AWGN_pw}
\end{equation}
where
\begin{equation}
f(\bldq) = \frac{j}{k} \log B^{(M)}(\bldx_0) - j \sum_{r=1}^{M} q_r \log x_{0,r} - (j-1) h(\bldq)
\label{eq:f_definition}
\end{equation}
and $g(\bldq)$ is given by (\ref{eq:g_definition}). We solve this constrained optimization problem using Lagrange multipliers. At the maximum, we must have
\[
\frac{\partial f(\bldq)}{\partial q_r} = \lambda \frac{\partial g(\bldq)}{\partial q_r}
\]
for all $r = 1,2,\cdots, M$, where $\lambda$ denotes the Lagrange multiplier. This yields
\begin{multline}
\frac{j}{k} \left[ \frac{1}{B^{(M)}(\bldx_0)} \sum_{s=1}^{M} \frac{\partial B^{(M)} (\bldx_0)}{\partial x_{0,s}} \frac{\partial x_{0,s}}{\partial q_r} \right] \\
- j \left( \sum_{s=1}^{M} \frac{q_s}{x_{0,s}} \frac{\partial x_{0,s}}{\partial q_r} + \log x_{0,r} \right) \\
- (j-1) \frac{\partial h(\bldq)}{\partial q_r} = \lambda \frac{\partial g(\bldq)}{\partial q_r} 
\end{multline}
which is equivalent to
\begin{multline}
j \sum_{s=1}^{M} \frac{\partial x_{0,s}}{\partial q_r} \left[ \frac{1}{k B^{(M)}(\bldx_0)} \frac{\partial B^{(M)}(\bldx_0)}{\partial x_{0,s}} - \frac{q_s}{x_{0,s}} \right] \\
- j \log x_{0,r} + (j-1) \log \left( \frac{q_r}{1-\sum_{s=1}^{M} q_s} \right) \\
= \lambda \left( 2r \sum_{s=1}^{M} s q_s - \alpha r^2 \right) \; .
\end{multline}
The term in square brackets is equal to zero for each $r=1,2,\cdots,M$ due to (\ref{eq:bldx0_q_eqn}); therefore this simplifies to (\ref{eq:lagrange_mult_theorem}) for each $r=1,2,\cdots,M$.
\end{proof}

\begin{figure}
\begin{center}\includegraphics[%
  width=1.0\columnwidth,
  keepaspectratio]{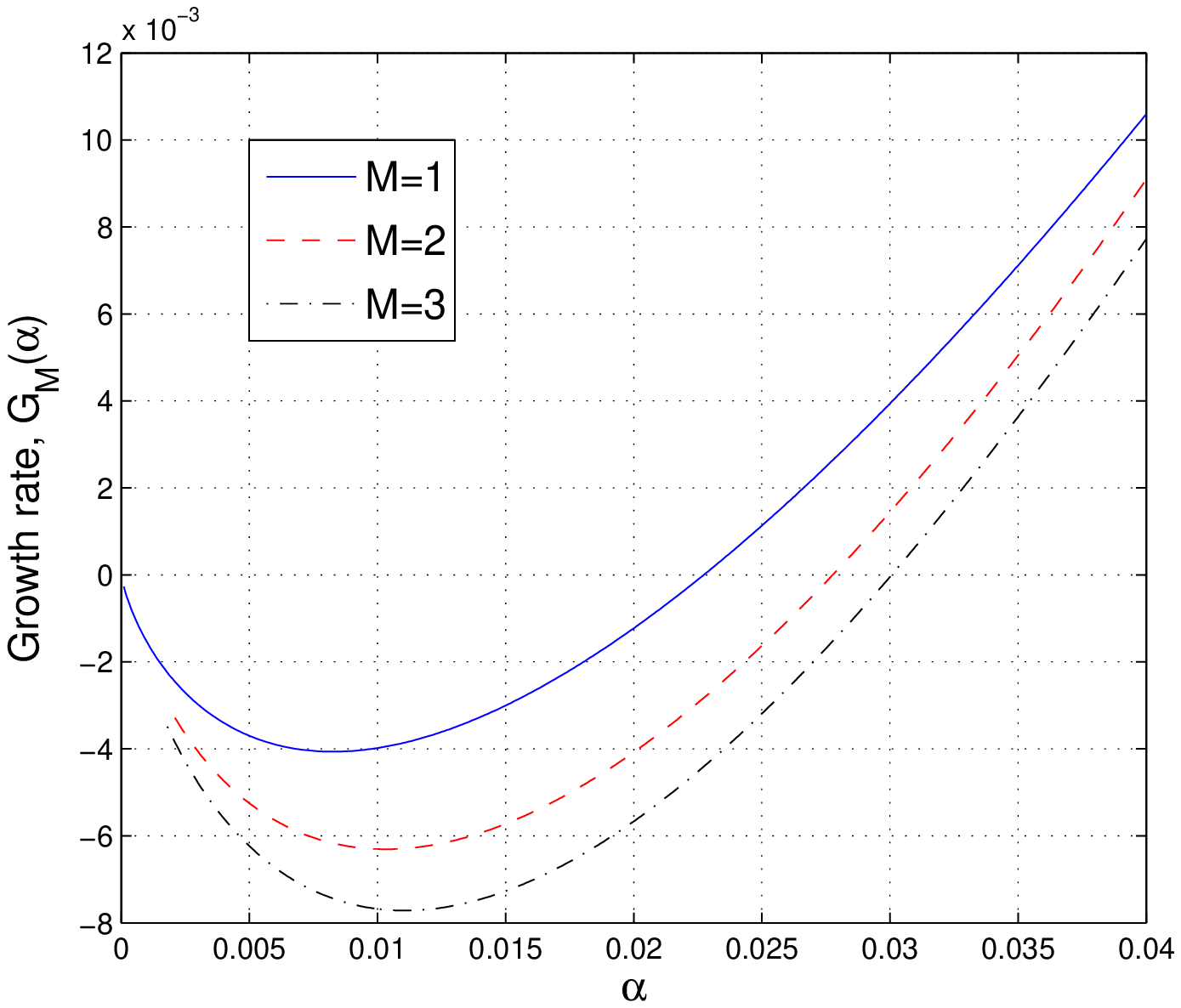}
\end{center}
\caption{\label{cap:Gallager_36_ensemble} Growth rate of the degree-$M$ AWGN-pseudoweight distribution for the $(3,6)$-regular ensemble ($M \le 3$).}
\end{figure}

Note that for the case $M=1$, the maximization in (\ref{eq:optimization_for_AWGN_pw}) is trivial and therefore the solution may be obtained directly from (\ref{eq:f_definition}) as 
\begin{multline}
G_1(\alpha) = \frac{j}{k} \log \left[ \frac{ \left(1+x_0 \right)^{k} + \left( 1-x_0 \right)^{k}}{2} \right] \\
- j \alpha \log x_{0} - (j-1) h(\alpha)
\label{eq:growth_rate_M1}
\end{multline}
where $x_0$ is the unique positive real solution to the equation
\begin{multline}
\alpha \left[ \left( 1+x_0 \right)^{k} + \left( 1-x_0 \right)^{k} \right] = \\
x_0 \left[ \left( 1+x_0 \right)^{k-1} - \left( 1-x_0 \right)^{k-1} \right] \; .
\label{eq:x0_M1}
\end{multline}
%\begin{equation}
%\alpha \left[ \left( 1+x_0 \right)^{k} + \left( 1-x_0 \right)^{k} \right] = x_0 \left[ \left( 1+x_0 \right)^{k-1} - \left( 1-x_0 \right)^{k-1} \right] \; .
%\label{eq:x0_M1}
%\end{equation}
Note that $G_1(\alpha)$ is simply the growth rate of the weight distribution in this case, originally obtained in \cite{gallager63:low-density}. Also, this solution may be regarded as a special case of Theorem \ref{thm:growth_rate} where the solution for $\lambda$ via (\ref{eq:lagrange_mult_theorem}) is redundant.
%Finally, note that a very similar analysis will yield the growth rate of the BSC-pseudoweight -- only the function $g(\bldq)$ changes. 

\begin{figure}
\begin{center}
\includegraphics[%
  width=1.0\columnwidth,
  keepaspectratio]{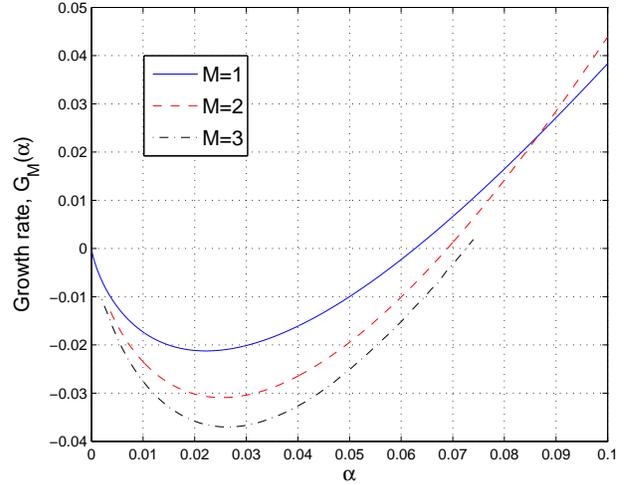}
\end{center}
\caption{\label{cap:Gallager_48_ensemble} Growth rate of the degree-$M$ AWGN-pseudoweight distribution for the $(4,8)$-regular ensemble ($M \le 3$).} 
\end{figure}

\section{Examples and Discussion}
\label{sec:example}
In this section the growth rates of the AWGN-pseudoweight of two example LDPC code ensembles are evaluated using the solution of Theorem \ref{thm:growth_rate}. The growth rate curves for the $(3,6)$-regular LDPC code ensemble and for the $(4,8)$-regular LDPC code ensemble are shown in Figures \ref{cap:Gallager_36_ensemble} and \ref{cap:Gallager_48_ensemble} respectively. Note that $0 < \alpha^*_1 < \alpha^*_2 < \alpha^*_3$ for both ensembles. It is worthwhile to note some distictions between the present analysis and that of \cite[Corollary 50]{KV-IEEE-IT}. In \cite[Corollary 50]{KV-IEEE-IT}, it is proved that $(j,k)$-regular ensembles with $j \ge 3$ have a ratio of minimum AWGN-pseudoweight to block length $n$ which decreases to zero asymptotically as $n \rightarrow \infty$. This result is not in conflict with the results of Figures \ref{cap:Gallager_36_ensemble} and \ref{cap:Gallager_48_ensemble}. The detrimental pseudocodewords of \cite[Corollary 50]{KV-IEEE-IT} are derived from the ``canonical completion'' \cite[Definition 46]{KV-IEEE-IT} and, asymptotically, have AWGN-pseudoweight \emph{sublinear} in the block length -- therefore, these pseudocodewords do not appear in the present analysis. Also, note that the analysis of \cite[Corollary 50]{KV-IEEE-IT} takes the limit $M \rightarrow \infty$ prior to (or jointly with) the limit $n \rightarrow \infty$, in contrast to the present analysis which takes the limit $n \rightarrow \infty$ for finite $M$. Finally, the result of \cite[Corollary 50]{KV-IEEE-IT} is concerned with \emph{minimum} AWGN-pseudoweight and not with the multiplicities of the corresponding pseudocodewords.
%\section{Conclusion}
%A solution was presented for the asymptotic growth rate of the AWGN-pseudoweight distribution of regular LDPC ensembles for a selected graph cover degree $M \ge 1$. The evaluation requires solution of a $(2M+1) \times (2M+1)$ system of equations. Simulation results were presented for two example regular LDPC code ensembles.

%----------------------------------------------------------------------

\bibliographystyle{IEEEtran}

\end{document}